\documentclass[letterpaper, 10 pt, conference]{ieeeconf}  

\IEEEoverridecommandlockouts                              

\usepackage{graphics} 
\usepackage{epsfig} 
\usepackage{mathptmx} 
\usepackage{mathtools} %
\usepackage{booktabs}
\usepackage{array}
\usepackage{xcolor}

\usepackage{amsthm}
\usepackage{graphics} 
\usepackage{epsfig} 
\usepackage{times} 
\usepackage{amsmath} 
\usepackage{amssymb} 
\usepackage{cite}
\usepackage{hyperref}
\usepackage{dsfont}
\usepackage{enumerate}
\usepackage{bbm}
\usepackage{algorithm}
\usepackage{algpseudocode}

\theoremstyle{plain}
\newtheorem{theorem}{Theorem}

\newtheorem{lemma}{Lemma}
\newtheorem{proposition}{Proposition}

\theoremstyle{definition}
\newtheorem{definition}{Definition}

\title{\LARGE \bf
Associative Memory System via Threshold Linear Networks
}

\author{Qin (Eric) He, Jing Shuang (Lisa) Li
\thanks{The authors are with the Department of Electrical Engineering and Computer Science, University of Michigan.  
 {\tt\small \{ericheq@umich.edu, jslisali\}@umich.edu}}
}

\begin{document}
\maketitle

\begin{abstract}

Humans learn and form memories in stochastic environments. Auto-associative memory systems model these processes by storing patterns and later recovering them from corrupted versions. Here, memories are learned by associating each pattern with an attractor in a latent space. After learning, when (possibly corrupted) patterns are presented to the system, latent dynamics facilitate retrieval of the appropriate uncorrupted pattern. In this work, we propose a novel online auto-associative memory system. In contrast to existing works, our system  supports sequential memory formation and provides formal guarantees of robust memory retrieval via region-of-attraction analysis. We use a threshold-linear network as latent space dynamics in combination with an encoder, decoder, and controller. We show in simulation that the memory system successfully reconstructs patterns from corrupted inputs.
\end{abstract}

\section{Introduction} \label{sec:intro}
An auto-associative memory system stores patterns and later reconstructs them from partial or noisy versions of those same patterns (i.e., pattern completion \cite{hopfield1982neural}). 
For example, under poor lighting conditions, the human brain can reconstruct a familiar cat well enough for recognition. 
Computational neuroscientists and computer scientists often model such systems using a dynamical latent space. During learning, patterns are associated with fixed-point attractors in latent space through an encoder and decoder; during inference, the memory system retrieves the appropriate pattern by leveraging latent dynamics. 
Previous work on associative memory systems includes energy-based models such as Hopfield networks \cite{hopfield1982neural, krotov2023new} and neural network-related models \cite{ramsauer2020hopfield}. These frameworks provide principled mechanisms for associative memory, with applications in behavioral modeling \cite{he2025adaptive} and machine learning \cite{krotov2025modern}.
However, classical Hopfield networks can have spurious attractors \cite{hertz2018introduction}
, while neural networks lack formal guarantees for robust pattern retrieval. Moreover, most existing work requires full prior knowledge of stored patterns and does not support online learning as humans do.

We propose a novel online auto-associative memory system with robustness guarantees. 
The system consists of a threshold-linear network (TLN) latent space (section \ref{sec:TLN}), a controller (section \ref{subsec:controller}), an encoder, and a decoder (section \ref{subsec:learning}). 
Figure~\ref{fig:TLN_arch} shows the full architecture. The encoder maps patterns to the latent space. Latent dynamics are modulated by the controller for memory formation and retrieval. The decoder then maps latent states back to pattern space. 

During learning (Fig.~\ref{fig:TLN_arch}A), we sequentially present the system with new patterns to memorize (i.e., online learning). When a new pattern is presented, the decoder output does not match this pattern; this mismatch triggers the controller and induces a transition to a new attractor. Once the latent dynamics converge to the new attractor, the encoder and decoder are updated to match the pattern to the attractor.

During inference (Fig.~\ref{fig:TLN_arch}B), noisy patterns are presented to the system sequentially. Each noisy pattern is mapped through the encoder to a target state. The controller then drives the latent dynamics to the target state. When this target state lies in the region of attraction (ROA) of the relevant attractor, the decoder is guaranteed to reconstruct the correct pattern from that attractor. We rigorously characterize the ROA of the latent dynamics and use this to provide robustness guarantees (section \ref{sec:ROA_forward}); this is, to our knowledge, the first such characterization for TLN dynamics. We then validate our analyses in simulation (section \ref{sec:simulations}).

\textbf{Notation:} Let $\mathbb{S}^n \subseteq \mathbb{R}^{n \times n}$ denote the set of symmetric matrices. For $M \in \mathbb{S}^n$, we use $M \succeq 0$ and $M \succ 0$ denote positive semidefinite and positive definite, respectively. Let $\mathbf{1}_n \in \mathbb{R}^n$ denote the all-ones vector and $\mathbf{0}_n \in \mathbb{R}^n$ denote the all-zeros vector. We use $\mathbf{I}[\cdot]$ to denote the indicator function, which takes the value $1$ when its argument is true and $0$ otherwise. We denote an equilibrium by \(x^e\). When the equilibrium is an attractor or a saddle point, we write it as \(x^*\) or \(x^\dagger\), respectively.


\begin{figure*}[htbp] 
    \centering
    \includegraphics[width=0.9\textwidth]{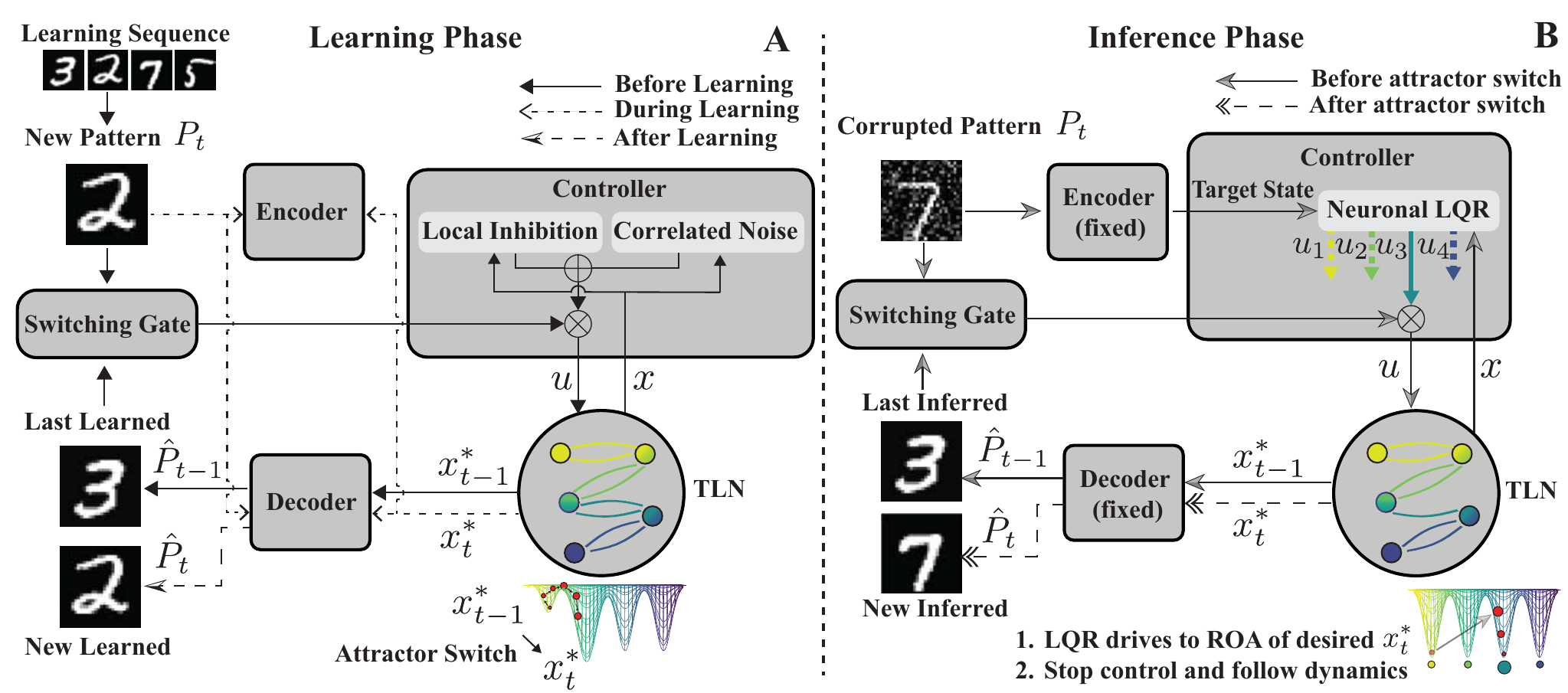} 
    \caption{Proposed auto-associative memory system. \textbf{A.} During learning, a new pattern is presented to the model, and a switching signal activates the controller to induce attractor switch and form new mappings (memories) in the encoder and decoder. \textbf{B.} Inference phase. The encoder maps a noisy pattern through the learned mappings to a target state. The feedback controller drives the latent dynamics toward the target state and is then turned off so the TLN settles into the corresponding attractor. The decoder then maps the attractor to the corresponding uncorrupted pattern.}
    \label{fig:TLN_arch}
\end{figure*}

\section{Threshold Linear Networks} \label{sec:TLN}

The proposed memory system uses threshold-linear network (TLN) dynamics as its latent space dynamics. TLNs have been used to study neural dynamics, such as firing-rate dynamics, are computationally tractable, and support fixed-point attractors.\cite{hahnloser2000permitted, morrison2016diversity}. In this section, we discuss TLNs and their equilibria. The standard TLN dynamics are:
\begin{align}\label{eq:basic_tln}
    \dot{x} = -x + [y]_{+}, \quad y := Wx+\theta
\end{align}
where $x \in \mathbb{R}^n$ is the state vector, $W \in \mathbb{S}^n$ is the connectivity matrix, and $\theta \in \mathbb{R}^n$ is the bias. The threshold-nonlinearity (i.e., rectifier) is defined as $[y]_+ := \max\{y,0\}$. We now introduce relevant definitions for TLN dynamics.

\begin{definition}
    \label{def:eq_supp}
    The \textit{support of an equilibrium} ${x^{e}}$ is the subset of active states, written as
        $\operatorname{supp} {x^{e}} := \{\, i \mid x^{e}_i > 0 \,\}$.
\end{definition}
\begin{definition}
    Let $[n]:=\{1,\dots,n\}$ and $\sigma\subseteq[n]$. The \textit{cell} associated with support $\sigma$ is:
    \begin{equation}
C_\sigma := \Big\{\, x\in\mathbb{R}^n \ \Big|\  y_i(x)\ge 0 \ \forall i\in\sigma,\ \ y_k(x)\le 0 \ \forall k\notin\sigma \,\Big\}.
    \end{equation} 
    For \(i=1,\dots,n-1\), we call \(C_{\{i,i+1\}}\) a double-support cell. For \(i=2,\dots,n-1\), we call \(C_{\{i-1,i,i+1\}}\) a triple-support cell.
\end{definition}

Our memory system uses a sub-class of TLNs called \textit{chain-structured TLNs (CSTLNs)}, which support multiple attractors by construction. In this paper, we will refer to a locally asymptotically stable equilibrium as an \textit{attractor}\footnote{In literature~\cite{morrison2016diversity}, this is also called a \emph{stable fixed point}}.

\begin{definition}\label{def:cstlns}
A \textit{CSTLN} is ~\eqref{eq:basic_tln}, with \(W\) and \(\theta\) defined as
\begin{equation}\label{eq:W}
\begin{aligned}
W_{k\ell}=
\begin{cases}
0, & k=\ell,\\
-1+\epsilon, & |k-\ell|=1,\\
-1-\delta, & |k-\ell|>1,
\end{cases}
\qquad
\begin{aligned}[t]
&\epsilon\in(0,1),\ \delta>0,\\
&\theta=c\,\mathbf{1}_n,\ c>0.
\end{aligned}
\end{aligned}
\end{equation}
\end{definition}
When parameters $\epsilon$, $\delta$, and $c$ are appropriately chosen, the CSTLN is competitive and nondegenerate; then, each cell of the CSTLN contains at most one equilibrium~\cite{morrison2016diversity}. We now present some basic results on the structure of these equilibria.

\begin{proposition}\label{rm:support}
For CSTLNs, the support of an equilibrium $x^{e}$ coincides with the support of the cell containing it.
\end{proposition}
\begin{proof}
    From \eqref{eq:basic_tln}, we have $x^{e}=[y(x^{e})]_+$. Therefore,
\begin{equation}\label{eq:relu_support}
x_i^e>0 \iff y_i(x^e)>0,\qquad x_i^e=0 \iff y_i(x^e)\le 0.
\end{equation}
By ~\cite[Cor.~5.5]{morrison2016diversity},
each cell of a CSTLN contains at most one equilibrium and \textit{no} equilibrium lies on its boundary.
\end{proof}

For the remainder of the paper, we will denote both the support of an equilibrium and the support of the cell containing it by \(\sigma\). Then, let $\bar{\sigma}=[n]\setminus\sigma$ denote the off-support states.
We can reorder states such that on-support states are followed by off-support states, i.e., $\tilde x = \Pi x=  [x_{\sigma},x_{\bar\sigma}]^\top$ where $\Pi \in \mathbb{R}^{n\times n}$ is a permutation matrix. The dynamics restricted to a cell with support $\sigma$ are:
\begin{equation}
\label{eq:subsys_on_supp}
\frac{d\tilde x}{dt}
= \tilde A_\sigma\tilde x + \begin{bmatrix} \theta_{\sigma} \\ 0 \end{bmatrix},
\quad
\tilde A_\sigma =
\begin{bmatrix}
- I_{|\sigma|} + W_{\sigma\sigma} & W_{\sigma\bar{\sigma}} \\
0 & - I_{n-|\sigma|}
\end{bmatrix},
\end{equation}
where $W_{\sigma\sigma}$ is obtained by keeping only the rows and columns of $W$ in $\sigma$,
$W_{\sigma\bar{\sigma}}$ is obtained by keeping rows in $\sigma$ and columns in $\bar{\sigma}$,
and $\theta_\sigma$ is $\theta$ restricted to states in $\sigma$. 
Since each cell of the CSTLN contains at most one equilibrium, the equilibrium $x_e$ of permuted system\eqref{eq:subsys_on_supp} lies in the cell defined by $\sigma$ if and only if the following conditions hold:
\begin{subequations}\label{eq:cell_cond}
\begin{align}
y_i(x^e) &> 0, \quad \forall i\in\sigma, \label{eq:on-cond}\\
y_k(x^e) &< 0, \quad \forall k\notin\sigma. \label{eq:off-cond}
\end{align}
\end{subequations}

CSTLNs contain two types of equilibria: attractors (which we use to store memories) and saddle points (which we use in robustness analysis). We now discuss this further.

\begin{definition}
Let $S \subseteq [n] := \{1,2,\ldots,n\}$, and let $j \in [n]\setminus S$. 
The \textit{distance} from $j$ to $S$ is
$\operatorname{dist}(j,S) := \min_{s\in S} |j-s|$.
\end{definition}

\begin{lemma}\label{lemma:stb_eq}
    For each integer $i=1,\dots,n-1$, double-support cell $C_{\{i,i+1\}}$ contains a unique attractor $x^*$  where
    \begin{equation}\label{eq:adj_sln}
    \begin{aligned}
    x_k^{*} &=
    \begin{cases}
    \frac{c}{2-\epsilon}, & k=i\ \text{and}\ k=i+1,\\[8pt]
    0, & \text{otherwise.}
    \end{cases}\\[6pt]
    \end{aligned}
    \end{equation}
    
\end{lemma}

\begin{proof}

First, following~\cite[Lem.~5.2]{morrison2016diversity}, we compute the equilibrium \(x^e_\sigma\) in each cell \(C_\sigma\) by restricting the dynamics to the subsystem on \(\sigma\) and setting the remaining states \(x^e_{\bar{\sigma}}\) to zero. Applying this to \(C_{\{i,i+1\}}\) and using~\eqref{eq:subsys_on_supp} yields:
\begin{equation}\label{eq:attractor_sol}
    x_\sigma^e \;=\; (I - W_{\sigma\sigma})^{-1}\theta_\sigma
\;=\; \frac{c}{2-\epsilon}\,\mathbf{1}_{|\sigma|},
\quad
x_{\bar{\sigma}}^e \;=\; \mathbf{0}_{|\bar{\sigma}|}.
\end{equation}
Next, we show that the equilibrium lies in \(C_{\{i,i+1\}}\). From~\eqref{eq:attractor_sol}, the on-support states satisfy~\eqref{eq:on-cond},
\begin{equation}
    x^e_\sigma = [y_\sigma(x^e)]_+, \quad x^e_\sigma=\frac{c}{2-\epsilon}>0 \Rightarrow y_\sigma(x^e) >0.
\end{equation}
For the off-support states, consider \(k=i-1\) (the case \(k=i+2\) is similar). Substituting~\eqref{eq:attractor_sol} into $y_{i-1}(x^e)$ gives
\begin{align}
    &y_{i-1}(x^e) = (Wx^e+\theta)_{i-1} \notag\\
    &= (-1+\epsilon)x_i+(-1-\delta)x_{i+1}+c=\frac{-c\delta}{2-\epsilon} <0
    \label{eq:off_exm}
\end{align}
For any state $k$ with $\operatorname{dist}(k,\{i,i+1\})>1$, $y_k(x^e)$ becomes more negative. Therefore, $y_k(x^e) < 0$ for all $k\notin\sigma$, so the off-support condition~\eqref{eq:off-cond} holds. Thus, \(x^e\in C_{\{i,i+1\}}\). We then show that $x^e$ is an attractor $x^*$. For \(C_{\{i,i+1\}}\), the matrix \(-I_{|\sigma|}+W_{\sigma\sigma}\) in~\eqref{eq:subsys_on_supp} becomes
\begin{equation}
- I_{|\sigma|} + W_{\sigma\sigma}=
\begin{bmatrix}
-1 & -1+\epsilon\\
-1+\epsilon & -1
\end{bmatrix},
\end{equation}
with eigenvalues \(\lambda_1=-2+\epsilon\) and \(\lambda_2=-\epsilon\). Under the CSTLN parameter constraints in Definition~\ref{def:cstlns}, both eigenvalues are negative. The off-support Jacobian \(-I_{n-|\sigma|}\) in~\eqref{eq:subsys_on_supp} has eigenvalue \(-1\) only. Therefore, the system Jacobian \(\tilde A_{\sigma}\) for the dynamics in \(C_{\{i,i+1\}}\) is Hurwitz, which implies the equilibrium \(x^e\) is an attractor.
\end{proof}


\begin{lemma}[Triple-support cell saddle]
\label{lemma:tpl_sln}
For each integer $i=2,\dots,n-1$, triple-support cell $C_{\{i-1, i,i+1\}}$ has a unique hyperbolic saddle point $x^{\dagger}$ where
\begin{equation}\label{eq:tpl_sln}
m:=\frac{c\varepsilon}{\Delta},\qquad
n:=\frac{c(\delta+2\varepsilon)}{\Delta},\qquad
\Delta:=\delta+4\varepsilon-2\varepsilon^2,
\end{equation}
\begin{equation}\label{eq:tpl_sln_compact}
x_k^\dagger=
\begin{cases}
m, & k=i-1 \text{ or } k=i+1,\\
n, & k=i,\\
0, & \text{otherwise.}
\end{cases}
\end{equation}
\end{lemma}

\begin{proof}
Similar to proof for Lemma~\ref{lemma:stb_eq}. An equilibrium of the cell $C_{\{i-1, i,i+1\}}$ is given by
\begin{equation}
x^e_\sigma=(I - W_{\sigma\sigma})^{-1}\theta_\sigma=[m,n,m]^\top,\qquad x^e_{\bar \sigma}=0.
\label{eq:tpl_exp_sln}
\end{equation}
We now show that $x^e_\sigma$ lies in \(C_{\{i-1,i,i+1\}}\). Under CSTLN parameter constraints in Definition~\ref{def:cstlns}, we have $\Delta>0$, and hence $m>0$ and $n>0$. Therefore, on-support states satisfies~\eqref{eq:on-cond}
\begin{equation}
    x^e_\sigma = [y_\sigma(x^e)]_+, \quad x^e_\sigma>0 \Rightarrow y_\sigma(x^e) >0.
\end{equation}
For off-support states, consider $k=i-2$ (similar for $k=i+2$)
\begin{align*}
&y_{i-2}(x^e) = (-1+\epsilon)x_{i-1}^e + (-1-\delta)x_{i}^e + (-1-\delta)x_{i+1}^e +c\\
&= c - (2+\delta-\varepsilon)m - (1+\delta)n = -\,\frac{c(\delta^2 + 3\delta\varepsilon + \varepsilon^2)}{\Delta} < 0.
\end{align*}
For any state $k$ with $\operatorname{dist}(k,\{i-1,i,i+1\})>1$, $y_k(x^e)$ becomes more negative. Therefore, the off-condition~\eqref{eq:off-cond} holds. Hence, $x^e \in C_{\{i-1, i,i+1\}}$. Then, we show that $x^e$ is a hyperbolic saddle point $x^\dagger$. For $C_{\{i-1, i,i+1\}}$, matrix $- I_{|\sigma|} + W_{\sigma\sigma}$ in~\eqref{eq:subsys_on_supp} becomes
\begin{equation}
    - I_{|\sigma|} + W_{\sigma\sigma} = \begin{bmatrix}
-1 & -1+\varepsilon & -1-\delta\\
-1+\varepsilon & -1 & -1+\varepsilon\\
-1-\delta & -1+\varepsilon & -1
\end{bmatrix},
\label{eq:triple_jss}
\end{equation}
with eigenvalues \(\lambda_1=\delta\) and \(\lambda_{2,3}=-1+\frac{b\pm\sqrt{b^2+8a^2}}{2}<0\), where \(a=-1+\epsilon\) and \(b=-1-\delta\). So~\eqref{eq:triple_jss} has one eigenvalue with positive real part and two with negative real parts. Off-support Jacobian $-I_{n-|\sigma|}$ in~\eqref{eq:subsys_on_supp} has eigenvalue $-1$. Therefore, the equilibrium \(x^e\) is a hyperbolic saddle point. 
\end{proof}

\section{Controller and Learning Rules}
Our memory system consists of a CSTLN (described above) in addition to a decoder, encoder, and controller. We now describe these additional components.

\subsection{Controller and action trigger}\label{subsec:controller}

Controller action is triggered by a mismatch between input pattern $P_t \in \mathbb{R}^{d}$ and output pattern $\hat{P_t} \in \mathbb{R}^{d}$. To measure mismatch, let \(s(t)\) denote cosine similarity\footnote{The memory system can work equally well with other similarity measures such as Euclidean distance.} between patterns $P_t$ and $\hat{P}_t$. When similarity drops below some threshold, we generate a trigger signal for the controller:
\begin{subequations}\label{eq:switch_dyn}
\begin{align}
\gamma(t) &= g\!\left(m\big(s_{\mathrm{th}}-s(t)\big)\right), \hspace{0.5em}
g(x)=\frac{1}{1+e^{-x}} \hspace{1em} \text{(trigger)} \label{eq:switch_dyn_a} \\
\dot q(t) &= \frac{\gamma(t)(1-q(t))}{\tau_q} \hspace{9.7em} \text{(latch)}
\label{eq:switch_dyn_b}\\
\dot T(t) &= q(t) \hspace{13.7em} \text{(timer)} \label{eq:switch_dyn_c}\\
w(t) &= g\!\left(\beta\big(T(t)-H\big)\right) 
\hspace{4.5em} \text{(timer threshold)} \label{eq:switch_dyn_d}\\
\dot G(t) &= \frac{(1-G(t))\,\gamma(t)\,(1-q(t))}{\tau_r}-\frac{G(t)\,w(t)}{\tau_d} \hspace{0.5em}
{\text{(pulse)}} \label{eq:switch_dyn_e}
\end{align}
\end{subequations}
where \(m\) is a positive constant chosen large enough to ensure fast mismatch detection, mismatch threshold \(s_{\mathrm{th}}\) is set near \(0.9\), \(\tau_q,\tau_r,\tau_d\) are time constants typically selected in the range \(2\)–\(10\), and \(H\) sets the pulse duration and should be large enough to trigger a successful transition without causing persistent switching. The pulse \(G(t)\) is then clipped to \([0,1]\). Intuitively, \(G(t)\) rises when a new pattern is detected, remains active for duration \(H\) to trigger the controller, then decays after sufficient control input has been injected. We now discuss the separate controllers for learning and inference.

\subsubsection{Controller for learning}
During learning, the controller injects input $u_{\mathrm{learn}}$ into the latent dynamics~\eqref{eq:basic_tln} as follows:
\begin{subequations} \label{eq:tln_learning}
\begin{align}
\dot x(t) \;&=\; -x(t) + \big[\,Wx(t)+\theta \;+\; rG(t)u_{\mathrm{learn}}(t)\big]_+ ,\\
u_{\mathrm{learn}}(t) &= W_{\mathrm{inh}}x(t)+\xi(x)\big.
\end{align}
\end{subequations}
where $W_{\mathrm{inh}}x$ represents local inhibition ($W_{\mathrm{inh}} \in \mathbb{R}^{n\times n}_{< 0}$), $\xi(x)$ represents correlated noise, and $r$ is a positive scalar. 
In general, in the absence of inputs, the latent state will settle at an attractor.
When the controller is triggered, $W_{\mathrm{inh}}$ is updated.
This reshapes the latent dynamics so that the state can be easily induced (with correlated noise) to another attractor.
For example, to move the state from the attractor with support \(\sigma^i=\{i,i+1\}\) to the one with support \(\sigma^{i+1}=\{i+1,i+2\}\), we set
\begin{equation}\label{eq:local_inh}
[W_{\mathrm{inh}}]_{\sigma^i\sigma^i}=
\begin{bmatrix}
0 & -c_{\mathrm{inh}}\\
-c_{\mathrm{inh}} & 0
\end{bmatrix},
\end{equation}
where $c_{\mathrm{inh}} \in \mathbb{R}_{>0}$. Correlated noise is
\begin{subequations}\label{eq:noise_model}
\begin{align}
\xi(x) &= \kappa \tanh\!\big(a(t)\big)\,k\big(x\big), \label{eq:noise_model_a}\\
\dot a(t) &= -\frac{1}{\tau}a(t)+\sqrt{\frac{2}{\tau}}\,\zeta(t). \label{eq:noise_model_b}
\end{align}
\end{subequations}
where $k(x)\in\mathbb{R}^n$ is a spatial weighting profile across states, centered on the current support and decaying exponentially away from it; $a(t)$ is an Ornstein–Uhlenbeck (OU) process with time constant $\tau$; \(\zeta(t)\) is Gaussian white noise. Saturation \( \kappa\tanh(a(t))\) is used to avoid large perturbations. Parameters \(c_{\mathrm{inh}}\), \(r\), $\kappa$, and \(\tau\) should be chosen so that the attractor transition occurs within duration \(H\). 

\subsubsection{Controller for inference}
During inference, the encoder maps a (possibly corrupted) input pattern to a target state \(x_{\mathrm{tar}}\) near the desired attractor. The controller injects state feedback input $u_{\mathrm{fb}}$ into latent dynamics \eqref{eq:basic_tln} as follows:
\begin{subequations}\label{eq:inference_controller}
\begin{align}
\dot x(t) &= -x(t)+[Wx(t)+\theta]_+ + rG(t)u_{\mathrm{fb}}(t), \label{eq:inference_controller_a}\\
u_{\mathrm{fb}}(t) &= -K\bigl(x(t)-x_{\mathrm{tar}}\bigr). \label{eq:inference_controller_b}
\end{align}
\end{subequations}
We design a gain matrix $K$ so that the closed-loop system \eqref{eq:inference_controller} drives the state $x$ to the target state $x_\mathrm{tar}$. To do so, we use an LQR controller derived from the linearization (about $x_{\mathrm{tar}}$) of the closed-loop dynamics.
\begin{subequations}\label{eq:linearize_model}
\begin{align}
&\dot x-\dot x_{\mathrm{tar}}
= A\bigl(x-x_{\mathrm{tar}}\bigr)+rG(t)u_{\mathrm{fb}}, \label{eq:linearize_model_a}\\
&A= -I+DW,\quad
D_{ii}=\mathbf{I}_{\{w_i^\top x_{\mathrm{tar}}+\theta_i>0\}}. \label{eq:linearize_model_b}
\end{align}
\end{subequations} 
Here, $A\in\mathbb{R}^{n\times n}$ is the state matrix of the local linearized model and $D$ selects the positive components of the threshold nonlinearity. Let $B(t)=rG(t)$, we introduce an additional auxiliary dynamical system to compute the optimal LQR gain:
\begin{subequations}\label{eq:LQR_dynamics}
\begin{align}
\hspace{-0.5em}g_{\rm on} &= [(h-h_{\rm on})]_+, \,
g_{\rm off} = [(h_{\rm off}-h)]_+, \,
h = Wx_{\mathrm{tar}}+\theta,\\
\dot D &= \tfrac{1}{\tau_d}\!\left(\alpha(I-D)\operatorname{Diag}(g_{\rm on})-\beta D\operatorname{Diag}(g_{\rm off})\right),\\
\dot A &= \tfrac{1}{\tau_a}\left(-A+DW-I\right), \qquad
A_{\rm cl}=A-BK,\\
\dot P &= \tfrac{1}{\tau_p}\!\left(A_{\rm cl}^\top P+PA_{\rm cl}+Q+K^\top RK\right),\\
\dot K &= -\tfrac{1}{\tau_k}\left(RK-B^\top P\right),
\end{align}
\end{subequations}
where \(h_{\rm on}\) is an activation threshold chosen slightly above \(0\) while \(h_{\rm off}=0\). Constants \(\alpha\) and \(\beta\) are selected so that \(D(t)\) quickly converges to $D$ in~\eqref{eq:linearize_model_b}.
Time constants in~\eqref{eq:LQR_dynamics} are chosen to be small (around \(0.01\)) so that \(D(t)\), \(A(t)\), and \(P(t)\) converge faster than \(K(t)\), which in turn converges to the optimal LQR solution.
After duration $H$, the state reaches $x_\mathrm{tar}$ and the controller is turned off as \(G(t)\) approaches zero.
Then, the CSTLN returns to its autonomous dynamics \eqref{eq:basic_tln}. If the target state $x_{\mathrm{tar}}$ lies in the ROA of the intended attractor, then CSTLN dynamics will converge to it, and the decoder will recover the correct pattern.
\textbf{Remark:} We formulate the control mechanism as a dynamical system; this is a desirable property for biological plausibility, since biological neurons perform computations via dynamics. If biological plausibility is not a priority, one could equally use a standard, non-dynamical LQR setup.

\subsection{Encoder and decoder learning rules}\label{subsec:learning}
The encoder $W_E\in\mathbb{R}^{d \times n}$ maps input pattern $P \in \mathbb{R}^d$
to a target latent state via $x_{\mathrm{tar}}=W_E^\top P$.
The decoder $W_D\in\mathbb{R}^{n\times d}$ maps latent state $x$ to output pattern $\hat{P}$ via $\hat P=W_D^\top x$. 

During learning, when a new pattern $P$ is presented, the abovementioned controller mechanism moves the latent state to a new attractor $x^*$. Once this occurs, we update $W_E$ and $W_D$ to associate $x^*$ with $P$ . For the encoder, we only update the rows associated with support $\sigma$ of $x^*$. Let $F$ be the set of frozen rows (i.e., weights associated with previously learned patterns) and define $U:=\sigma \backslash F$. This gives the decomposition
\begin{equation}\label{eq:encoder_learn}
    x^\ast = W_E^\top P
      = \underbrace{W_{E,F}^\top}_{\text{frozen rows}}P
      + \underbrace{W_{E,U}^\top}_{\text{free rows}}P.
\end{equation}
We then choose the minimum-norm update for the free rows using the Moore--Penrose solution:
\begin{equation}\label{eq:weight_update}
W_{E,U}^\top
=
\frac{\bigl(x^* - W_{E,F}^\top P\bigr)P^\top}{\|P\|_2^2}.
\end{equation}
This ensures that after learning, new pattern $P$ is associated with $x^*$ without disrupting previously learned associations. An analogous column update process applies to the decoder. The overall memory system including latent dynamics, controller, encoder, and decoder is presented in Algorithm \ref{alg:memory}.

\begin{algorithm}[ht]
\caption{Proposed memory system}
\label{alg:memory}
\begin{algorithmic}
\State \textbf{Learning:} \textbf{for} each new pattern \(P\) \textbf{do:}
    \State \quad Compute cosine similarity $s(t)$ and control trigger \eqref{eq:switch_dyn}
    \State \quad Apply control \eqref{eq:tln_learning} to drive latent state $x$ to attractor $x^*$
    \State \quad Update encoder and decoder weights~\eqref{eq:weight_update}
\State \textbf{Inference:} \textbf{for} each corrupted input pattern \(P\) \textbf{do:}
    \State \quad Compute cosine similarity $s(t)$ and control trigger \eqref{eq:switch_dyn}
    \State \quad Compute target state \(x_{\mathrm{tar}}=W_E^\top P\)
    \State \quad Apply control \eqref{eq:inference_controller} to drive the latent state $x$ to $x_{\mathrm{tar}}$
    \State \quad Deactivate control after duration $H$, $x$ will settle to $x^*$
    \State \quad Compute the reconstructed pattern $\hat P=W_D^\top x^*$
\end{algorithmic}
\end{algorithm}    

\section{Region of attraction analysis} \label{sec:ROA}
The robustness of our memory system depends on the region of attraction (ROA) of the TLN attractors.
We now present ROA computations for TLNs.
In general, computing ROA for high-dimensional nonlinear systems is difficult \cite{topcu2008local}.
For TLNs, existing analyses require additional assumptions on network structure \cite{lienkaemper2022combinatorial, alvarez2026attractor}.
In contrast, we now present (i) a Lyapunov-based approximation of ROAs for general TLNs and (ii) a geometric method for exact calculation of ROAs for CSTLNs. We then use these to provide robustness guarantees for the proposed memory system.


\subsection{ROA certification for TLNs via SDP}
A common way to compute an inner ROA approximation is to find a Lyapunov function; 
this can be formulated as a semidefinite program (SDP) \cite{liao2022quadratic}. 
This method has been applied to threshold nonlinearities, leveraging quadratic constraints (QCs) derived from slope-restricted sector bounds\cite{noori2024complete}. 
However, \cite{noori2024complete} uses a global slope restriction, which is too conservative for our application; we now introduce a local slope restriction similar to~\cite{yin2021stability} and use this to find an appropriate Lyapunov function.
We first shift coordinates so that the attractor is at the origin. Define \(z\) such that:
\begin{equation}
     z = x - x^*, \qquad y^* := W x^* + \theta.
\end{equation}
The dynamics in the shifted coordinates are:
\begin{equation}
    \dot{z} = -z + [Wz +y^*]_{+} - [y^*]_{+} = -z + \phi(z),
    \label{eq:shifted_sys}
\end{equation}
where \(\phi(z) = [Wz + y^*]_{+} - [y^*]_{+}\). Consider the local ellipsoidal region $\mathcal{D}:=\{\, z \mid z^\top E z \le \alpha^2 \,\}$, where $E\in\mathbb{S}^n$, $E\succ 0$, and $\alpha>0$.
For componentwise analysis, let $v:=Wz+y^*$. Then, for all \(z \in \mathcal D\), each $v_i$ satisfies:
\begin{equation}
    v_i \in \left[ y^*_i - \alpha \sqrt{w_i^\top E^{-1} w_i},\; y^*_i + \alpha \sqrt{w_i^\top E^{-1} w_i} \right],
    \label{eq:ele_bds}
\end{equation}
where $w_i^\top$ is the $i$-th row of $W$. Let $[\underline v,\overline v]$ denote the componentwise bounds obtained by stacking the intervals in~\eqref{eq:ele_bds}. Equivalently, we can write $\phi(v)=[v]_+ - [y^*]_+$ for $v \in [\underline{v}, \overline{v}]$.  Since the shifted system~\eqref{eq:shifted_sys} has an equilibrium at \(z^*=0\), evaluate $(v, \phi(v))$ at $z^*$, we have:
\begin{subequations}\label{eq:star_defs}
\begin{align}
    v^* &:= W z^* + y^* = y^*, \label{eq:star_defs_a}\\
    \phi(v^*) &:= [W z^* + y^*]_+ - [y^*]_+ = 0. \label{eq:star_defs_b}
\end{align}
\end{subequations}
Since $\phi$ acts elementwise, for each $i$ we have
$
[\phi(v)]_i = [v_i]_+ - [y_i^*]_+,
$
each component $[\phi(v)]_i$ falls into one of the two cases shown in Figure~\ref{fig:local_relu}, depending on the sign of $y_i^*$.
\begin{figure}[htbp] 
    \centering
    \includegraphics[width=0.45\textwidth]{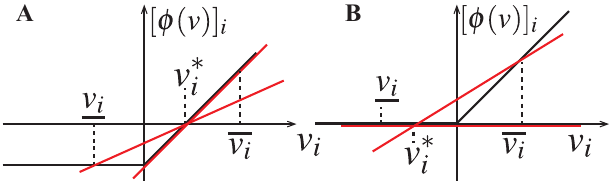} 
    \caption{Local sector bound for shifted system.}
    \label{fig:local_relu}
\end{figure}

Therefore, for $v\in[\underline{v},\overline{v}]$, $\phi(v)$ satisfies the tighter slope restriction $\frac{\Delta \phi(v)}{\Delta v}\in[S_{\alpha},S_{\beta}]$ around $(v^*,\phi(v^*))$. Following~\cite[Lemma~1]{yin2021stability}, if \( \lambda \in \mathbb{R}_{\ge 0}^{n} \) then
\begin{equation*}\label{eq:qc_main}
\begin{aligned}
&\begin{bmatrix}
    v - v^* \\
    \phi(v) - \phi(v^*)
\end{bmatrix}^\top
 \Psi^\top M(\lambda) \Psi
\begin{bmatrix}
    v - v^* \\
    \phi(v) - \phi(v^*)
\end{bmatrix}
\;\geq\; 0, \forall v \in [\underline{v}, \overline{v}], \\[8pt]
&\Psi := \begin{bmatrix}
        \operatorname{diag}(S_\beta) & -I \\
        -\operatorname{diag}(S_\alpha) & I
    \end{bmatrix},
M(\lambda) :=\begin{bmatrix}
        0 & \operatorname{diag}(\lambda) \\
        \operatorname{diag}(\lambda) & 0
    \end{bmatrix}.
\end{aligned}
\end{equation*}
With \(v=Wz+y^*\) and define a new matrix $M_{\kappa}$,
\[
M_{\kappa}\coloneqq
\begin{bmatrix} W & 0 \\ 0 & I \end{bmatrix}^{\!\top}
\Psi^{\top} M(\lambda)\,\Psi
\begin{bmatrix} W & 0 \\ 0 & I \end{bmatrix},
\]
the corresponding quadratic constraint in the $z$-coordinates is
\begin{equation}\label{eq:qc_zphi}
\begin{bmatrix} z \\ \phi(z) \end{bmatrix}^{\!\top}
M_{\kappa}
\begin{bmatrix} z \\ \phi(z) \end{bmatrix} \ge 0.
\end{equation}
After learning, for each attractor, we can solve an SDP to find a local ROA estimation and the maximum noise bound for its associated pattern, as given in the following theorem.

\begin{theorem}[SDP-based robustness certification]\label{thm:SDP} Given \(E \in \mathbb{S}^{n}\) with \(E \succeq 0\), and \(\alpha>0\), define the local domain \(\mathcal{D}=\{ z \in \mathbb{R}^{n} : z^\top E z \le \alpha^2 \}\). 
Define \(L:=Y\Sigma\), where $W_E = U\Sigma Y^\top$ is the compact singular value decomposition of encoder $W_E$ (i.e., with zero singular values discarded). Assume that the nonlinearity $\phi(z)$ in ~\eqref{eq:shifted_sys} with $A=-I$ and $B=I$ satisfies local QC ~\eqref{eq:qc_zphi} over \(\mathcal D\).  
If there exist matrix \(P\succ 0\), diagonal matrices \(S_0\succ 0\), \(S_1\succ 0\), and positive scalars $\epsilon$ such that 
\begin{subequations}\label{eq:SDP}
\begin{align}
&\begin{bmatrix}
A^\top P + P A & P B \\
B^\top P & 0
\end{bmatrix}
+ M_{\kappa}
\prec
\begin{bmatrix}
-\epsilon I & 0 \\
0 & 0
\end{bmatrix},
\label{eq:v_dot}\\[4pt]
&\frac{1}{\alpha^{2}}E \preceq P,
\label{eq:s_proc}\\[2pt]
&L^\top P L \preceq \frac{1}{r^{2}} I,
\label{eq:robust_lmi}
\end{align}
\end{subequations}
then
$
\mathcal{E}(P)=\{z\in\mathbb{R}^n : z^\top P z \le 1\}
$
is a local ROA estimate of attractor \(z^*\) for system~\eqref{eq:shifted_sys}. Furthermore, assume the input pattern to the memory system is a previously learned pattern corrupted with additive noise $\eta$. Then, as long as \(\|\eta\|_2 \le r\), the encoder will map the corrupted input pattern to a target state which lies in \(\mathcal{E}(P)\).
\end{theorem}

\begin{proof}
\(V(z)=z^\top P z\) is a Lyapunov function; pre- and post-multiplying~\eqref{eq:v_dot} by $[z^\top,\phi(z)^\top]^\top$ and applying~\eqref{eq:qc_zphi} yields $\dot V(z) \le -\varepsilon \|z\|_2^2$. 
In addition,~\eqref{eq:s_proc} implies that the set \(\{z : z^\top P z \le 1\}\) is contained in \(\{z : z^\top E z \le \alpha^2\}\) through S-procedure. Thus, \(\mathcal E(P)\) is contained in the local domain where the QC holds and is therefore forward invariant. Consequently, trajectories that starting in \(\mathcal E(P)\) converge to attractor \(z^*=0\) by the Lyapunov stability theory~\cite{khalil2002nonlinear}. Next, let \(z_0=W_E^\top\eta\) be the perturbation in $z$ coordinates.
Requiring \(z_0\in\mathcal E(P)\) for all \(\|\eta\|_2\le r\) implies
$
\eta^\top W_E P W_E^\top \eta \le 1,
\, \forall \eta:\|\eta\|_2\le r,
$
or, equivalently,
$
W_E P W_E^\top \preceq \frac{1}{r^2}I.
$
Since
$
W_E P W_E^\top
$
and
$
L^\top P L
$
have the same eigenvalues, this is equivalent to 
requiring \eqref{eq:robust_lmi}.
\end{proof}

\subsection{ROA certification for CSTLNs via LP}\label{sec:ROA_forward}
The SDP-based method yields hyperellipsoidal ROA approximations for general TLNs.
We now present a complementary method for CSTLNs specifically, which leads to a linear program (LP). We first construct forward-invariant sets, then show that each set is partitioned by a separating hyperplane into two regions; each region belongs to a different attractor's ROA. Each attractor's complete ROA can be assempled from pieces across these sets.

\begin{theorem}[Local forward-invariant set]\label{thm:forward_invariant}
For CSTLNs with \(n\ge 3\) and \(i=2,\ldots,n-1\), the following set
\begin{multline}\label{eq:RFI}
\mathcal{R}_{\mathrm{FI}}^{(i)}
:= \{\, x \in \mathbb{R}^n \mid
y_i(x)\ge \alpha,\ y_{i\pm1}(x)\le \beta,\\
y_k(x)\le 0 \ \forall\, k\notin \{i-1,i,i+1\} \,\},\quad \alpha,\beta\ge0 .
\end{multline}
is forward-invariant if there exist \(\alpha,\beta\in\mathbb{R}_{\ge 0}\) such that
\begin{subequations}\label{eq:feas}
\begin{align}
&\alpha \le c + 2(-1+\epsilon)\beta, \label{eq:alpha_1}\\
&(-1+\epsilon)\alpha + c \le \beta, \label{eq:alpha_2}\\
&\alpha \ge \frac{c}{1+\delta}. \label{eq:alpha_3}
\end{align}
\end{subequations}
\end{theorem}

\begin{proof}
We show that if~\eqref{eq:alpha_1}--\eqref{eq:alpha_3} hold, the vector field points inward on the boundary of \(\mathcal{R}_{\mathrm{FI}}^{(i)}\), and therefore \(\mathcal{R}_{\mathrm{FI}}^{(i)}\) is forward-invariant by Nagumo's theorem~\cite{nagumo1942lage}. Let $w_i^\top$ be the $i$-th row of $W$. On the boundary \(y_i(x)=w_i^\top x+c=\alpha\), differentiating \(y_i\) and restricting to \(\mathcal{R}_{\mathrm{FI}}^{(i)}\) in~\eqref{eq:RFI} gives:
\begin{align}
\dot y_i \;&=\; w_i^\top \dot x \;=\; -\,w_i^\top x \;+\; w_i^\top [y(x)]_{+}, \\
\label{eq:gidot}
\dot y_i \;&=\;c-\alpha + (-1+\epsilon)[y_{i-1}]_{+} + (-1+\epsilon)[y_{i+1}]_{+}.
\end{align}
Substituting \(y_{i\pm1}(x)\le\beta\) into~\eqref{eq:gidot}, \eqref{eq:alpha_1} implies that \(\dot y_i \ge 0\). Similarly, on the boundary $y_{i\pm1}=\beta$:
\begin{align}
\label{eq:y_cond_1}
\dot y_{i-1} \;=\; & c-\beta+(-1+\epsilon)[y_i]_{+} + (-1-\delta)[y_{i+1}]_{+},\\
\label{eq:y_cond_2}
\dot y_{i+1} \;=\; & c-\beta+(-1+\epsilon)[y_i]_{+} + (-1-\delta)[y_{i-1}]_{+}.
\end{align}
Substituting \(y_{i\pm1}(x)\le 0\) and \(y_i(x)\ge \alpha\) into~\eqref{eq:y_cond_1} and~\eqref{eq:y_cond_2}, it follows from~\eqref{eq:alpha_2} that \(\dot y_{i\pm1} \le 0\). On the boundary $y_k=0$, we consider the following 2 cases. If $\operatorname{dist}(j,\{i-1,i,i+1\}) \ge 2$:
    \begin{equation}\label{eq:node k}
    \dot y_k = c + (-1-\delta) ([y_{i-1}]_+ + [y_i]_+ + [y_{i+1}]_+).
    \end{equation}
    If $\operatorname{dist}(j,\{i-1,i,i+1\}) = 1$, consider \(y_{i+2}\) (similarly \(y_{i-2}\)),
    \begin{equation}\label{eq:node i_2}
    \dot y_{i+2} = c + (-1-\delta) ([y_{i-1}]_+ + [y_i]_+) + (-1+\epsilon)[y_{i+1}]_+.
    \end{equation}
    Substituting \(y_{i\pm1}(x)\le 0\) and \(y_i(x)\ge \alpha\) into~\eqref{eq:node k} and~\eqref{eq:node i_2}, it follows from~\eqref{eq:alpha_3} that \(\dot y_{k} \le 0\) and \(\dot y_{i+2} \le 0\). Therefore, the vector field points inward on every boundary of \(\mathcal{R}_{\mathrm{FI}}^{(i)}\). Since the TLN dynamics in~\eqref{eq:basic_tln} are globally Lipschitz and \(\mathcal{R}_{\mathrm{FI}}^{(i)}\) is a closed convex polyhedron, Nagumo's theorem implies that \(\mathcal{R}_{\mathrm{FI}}^{(i)}\) is forward invariant~\cite{nagumo1942lage,blanchini1999set}.
\end{proof}

\begin{proposition}\label{prop:feasibility}
For CSTLNs, conditions~\eqref{eq:alpha_1} -- \eqref{eq:alpha_3} are feasible only when plant parameter \(\epsilon \in [\tfrac{1}{2},1)\).
\end{proposition}
Next, we show that each forward-invariant set contains two attractors and one saddle point.
\begin{lemma}\label{lemma:RI_contains_Eq}
For CSTLNs with \(n \ge 3\) and integer \(i=2,\dots,n-1\), if a forward-invariant set \(\mathcal{R}_{\mathrm{FI}}^{(i)}\) exists for some \(\alpha,\beta\) satisfying~\eqref{eq:alpha_1}--\eqref{eq:alpha_3}, then \(\mathcal{R}_{\mathrm{FI}}^{(i)}\) contains the attractor in \(C_{\{i-1,i\}}\), the attractor in \(C_{\{i,i+1\}}\), and the saddle point in \(C_{\{i-1,i,i+1\}}\).
\end{lemma}
\begin{proof}
    We first show that the attractor \(x^* \in C_{\{i,i+1\}}\) lies in \(\mathcal{R}_{\mathrm{FI}}^{(i)}\). By Proposition~\ref{prop:feasibility} and~\eqref{eq:alpha_1}--\eqref{eq:alpha_2}, we have
    \begin{equation}\label{eq:reform_alpha}
        0 \le \alpha \le \alpha_{\max} = \frac{c(1-2\epsilon)}{1+2\epsilon^2-4\epsilon}
    \end{equation}
    By Lemma~\ref{lemma:stb_eq}, \(y_k^*=x_k^*=0\) for all \(k\notin\{i,i+1\}\), and \(y_{i-1}^*=x^*_{i-1}=0\le \beta\). Notice that,
    \begin{equation}\label{eq:yi_in_fi}
    y_i^*=x_i^*=\frac{c}{2-\epsilon}\ge \alpha_{\max}.
    \end{equation}
    Further, from~\eqref{eq:yi_in_fi} and~\eqref{eq:alpha_2},
    $
    \beta \ge c+(-1+\epsilon)\alpha \ge c+(-1+\epsilon)x_i^*=\frac{c}{2-\epsilon},
    $
    which implies
    \begin{equation}\label{eq:yi+1_in_fi}
    y_{i+1}^*=x_{i+1}^*=\frac{c}{2-\epsilon}\le \beta.
    \end{equation}
    Therefore, the attractor in \(C_{\{i,i+1\}}\) satisfies~\eqref{eq:RFI} and thus lies in \(\mathcal{R}_{\mathrm{FI}}^{(i)}\). Similar analysis holds for the attractor in $C_{\{i-1,i\}}$. For the saddle point \(x^\dagger \in C_{\{i-1,i,i+1\}}\), Lemma~\ref{lemma:tpl_sln} implies that \(y_k^*=x_k^\dagger=0\) for all \(k\notin\{i-1,i,i+1\}\). Together with the constraints in Definition~\ref{def:cstlns} and~\eqref{eq:yi_in_fi}--\eqref{eq:yi+1_in_fi}, we have
    \begin{equation}\label{eq:saddle_fi}
        x_i^\dagger = y^\dagger_i > \frac{c}{2-\epsilon}\ge\alpha_{\max}, \quad x_{i\pm1}^\dagger =y^\dagger_{i \pm 1} < \frac{c}{2-\epsilon}\le \beta.
    \end{equation}
     Therefore, the saddle point in $C_{\{i-1,i,i+1\}}$ lies in $\mathcal{R}_{\mathrm{FI}}^{(i)}$. 
\end{proof}

Next, we identify a hyperplane that divides \(\mathcal{R}_{\mathrm{FI}}^{(i)}\) into two sub-regions, each lying in the ROA of one attractor. 
\begin{lemma}\label{lemma:phi_dynamics}
For $i = 2,\dots,n-1$, let $A_{T}$\footnote{For notational simplicity, we drop the dependence of \(T\) and left eigenvector \(w\) on \(i\). In Theorem~\ref{thm:local_ROA}, we also drop the dependence of \(\mathcal{R}_{\mathrm{FI}}\) on \(i\).} be the Jacobian of dynamics in the triple-support cell $C_{\{i-1,i,i+1\}}$,
and let $w$ be a left eigenvector of $A_T$ associated with the positive eigenvalue $\delta$.
Define $\phi(x):=w^\top x$.
Then along any trajectory $x(t)$,
\begin{equation}\label{eq:Obs}
\dot\phi(x)=
\begin{cases}
\delta\,\phi(x), & x(t)\in C_{\{i-1,i,i+1\}},\\
-\phi(x), & x(t)\in C_{\{i\}}.
\end{cases}    
\end{equation}
\end{lemma}

\begin{proof}
    (1) In each triple-support cell, let $T =\{i-1,i,i+1\}$ for $i = 2, \cdots, n-1$, and \(J_{TT}:=-I_{|T|}+W_{TT}\), given explicitly in~\eqref{eq:triple_jss}. The dynamics in the triple-support cell with reordered states (as in~\eqref{eq:subsys_on_supp}) is
\begin{equation}
\label{eq:tri_cell_dyn}
\frac{d\tilde x}{dt}
= \tilde A_T\tilde x + \begin{bmatrix} \theta_{T} \\ 0 \end{bmatrix},
\quad
\tilde A_T =
\begin{bmatrix}
J_{TT} & W_{T\bar{T}} \\
0 & - I_{n-|T|}
\end{bmatrix}.
\end{equation}
Let \(\widetilde w^\top\) be a left eigenvector of \(\widetilde A_T\) corresponding to \(\delta\), and partition it as \(\widetilde w^\top=[\,w_u^\top\;\;v^\top\,]=[\,u_{i-1},u_i,u_{i+1},v^\top\,]\). Therefore,
\begin{equation}\label{eq:permuted}
w_u^\top J_{TT}=\delta\,w_u^\top, \qquad
w_u \in \mathrm{span}\{[1,0,-1]^\top\}.
\end{equation}
Substituting
\(\tilde w^\top \tilde A_T=\delta\,\tilde w^\top\) and
\(\tilde w^\top\!\begin{bmatrix}\theta_T\\0\end{bmatrix}=0\) into the following:
    \begin{equation}
    \dot{\tilde \phi}
    := \tilde w^\top \dot{\tilde x}
    = \tilde w^\top\!\left(\tilde A_T \tilde x + \begin{bmatrix}\theta_{T}\\[2pt] 0\end{bmatrix}\right) = \delta \tilde w^\top\tilde x = \delta\tilde\phi.
    \end{equation}
    In the original coordinates, this becomes:
    \begin{subequations}\label{eq:phi_backtransform}
    \begin{align}
    \dot\phi(t)
    &= \dot{\tilde\phi}(t)
    = \delta\,\tilde w^\top \tilde x(t)
    = \delta\,\tilde w^\top \Pi x(t), \label{eq:phi_backtransform_a}\\
    &= \delta\,(\Pi^\top \tilde w)^\top x(t)
    = \delta\,w^\top x(t)
    = \delta\,\phi(t). \label{eq:phi_backtransform_b}
    \end{align}
    \end{subequations}
    In \(C_{\{i\}}\) dynamics, the Jacobian is \(-I + D_{\{i\}}\widetilde W\), where \(D_{\{i\}}\) selects the row of \(\widetilde W\) corresponding to coordinate \(i\) and zeros out all others. Since \(u_i=0\) in~\eqref{eq:permuted}, we obtain \(\tilde w^\top D_{\{i\}}\widetilde W=0\). Therefore,
\begin{subequations}\label{eq:phi_dot_single}
\begin{align}
\dot{\tilde \phi}
&= \tilde w^\top \dot{\tilde x}
= \tilde w^\top\!\left[\left(-I + D_{\{i\}}\widetilde W\right)\tilde x
+ [0,\theta_i,0,\mathbf{0}_{n-3}^\top]^\top\right], \label{eq:phi_dot_single_a}\\
&= - \tilde w^\top \tilde x = -\tilde \phi. \label{eq:phi_dot_single_b}
\end{align}
\end{subequations}
    In the original coordinates, this becomes $\dot\phi(x)=-\,\phi(x)$.
\end{proof}

\begin{theorem}\label{thm:local_ROA}
Let \(\phi(x):=w^\top x\) as in~\eqref{eq:Obs}. Then, in each \(\mathcal{R}_{\mathrm{FI}}\),
\begin{enumerate}
    \item \(\{x:\phi(x)>0\}\cap\mathcal{R}_{\mathrm{FI}}\) is an inner approximation of the ROA of the attractor in \(C_{\{i-1,i\}}\).
    \item \(\{x:\phi(x)<0\}\cap\mathcal{R}_{\mathrm{FI}}\) is an inner approximation of the ROA of the attractor in \(C_{\{i,i+1\}}\).
    \item \(\{x:\phi(x)=0\}\cap\mathcal{R}_{\mathrm{FI}}\) is an invariant hyperplane separating the two ROAs.
\end{enumerate}
\end{theorem}

\begin{proof}
    (1) By Lemma~\ref{lemma:phi_dynamics}, \(\phi(x)=0\) implies \(\dot\phi(x)=0\). Therefore, the hyperplane \(\{\phi=0\}\) is invariant in \(\mathcal{R}_{\mathrm{FI}}\). Since the dynamics admit unique solutions, trajectories cannot cross this hyperplane. Additionally,~\cite[Lemma~2.1]{morrison2016diversity} shows that \(\mathcal{B}=\prod_{i=1}^n[0,\theta_i]\) is globally attracting for the competitive TLN. Thus, any trajectory starting in \(\mathcal{R}_{\mathrm{FI}}\) eventually enters and remains in the following compact set:
    \begin{equation}
        \Omega =\{x \in \mathbb{R}^n :\mathcal{R}_{\mathrm{FI}}\ \cap\ \mathcal{B}\}.
    \end{equation}
    Consider the following energy function on $\Omega$:
    \begin{equation}
        V(x) := \frac{1}{2}x^\top(I-W)x - \theta^\top x
    \end{equation}
    Since $W$ is symmetric, we have:
    \begin{subequations}\label{eq:V_deriv}
    \begin{align}
    \nabla V(x)
    &= x - \frac{1}{2}(W + W^\top)x - \theta = x-(Wx+\theta) \label{eq:V_deriv_b}\\
    \dot V
    &= \nabla V(x)^\top \dot x
    = (x-y)^\top(-x+[y]_+) \label{eq:V_deriv_c}\\
    &= \sum_{i=1}^n (x_i-y_i)\,(-x_i+[y_i]_+) \label{eq:V_deriv_d}.
    \end{align}
    \end{subequations}
    Then, consider the following two cases:
    \begin{enumerate}
    \renewcommand{\labelenumi}{(\alph{enumi})}
    \item If \(y_i>0\), then \((x_i-y_i)(-x_i+[y_i]_+) = -(x_i-y_i)^2 \le 0\).
    \item If \(y_i\le 0\), then \([y_i]_+=0\) and \((x_i-y_i)(-x_i) = -x_i^2+y_i x_i \le 0\), since \(x_i\ge 0\) on \(\Omega\) and \(y_i\le 0\).
    \end{enumerate}
    Therefore, \(\dot V(x)\le 0\) for all \(x\in\Omega\). Since \(\Omega\) is compact and forward invariant, LaSalle's invariance principle implies that every trajectory with \(x(0)\in\Omega\) converges to the largest invariant set in \(\{x\in\Omega:\dot V(x)=0\}\). By~\cite[Cor.~2.2]{morrison2016diversity}, all equilibria of CSTLNs lie in \(\mathcal{B}\). Together with Lemma~\ref{lemma:RI_contains_Eq}, this implies that \(\Omega\) contains two attractors and one saddle point, and every trajectory converges to an equilibrium as \(t\to\infty\). By Lemma~\ref{lemma:tpl_sln}, the saddle in \(C_{\{i-1,i,i+1\}}\) is hyperbolic. By~\cite[Sec.~5]{hirsch1989convergent}, the set of initial conditions converging to an unstable hyperbolic equilibrium is a lower-dimensional smooth manifold and therefore has Lebesgue measure zero. In \(\Omega\), this manifold is the hyperplane \(\phi(x)=0\). Therefore, for almost every initial condition in \(\Omega\), trajectories with \(\phi(x)>0\) converge to the attractor in \(C_{\{i-1,i\}}\), while trajectories with \(\phi(x)<0\) converge to the attractor in \(C_{\{i,i+1\}}\).
\end{proof}

\begin{figure}[h]
    \centering
    \includegraphics[width=0.9\linewidth]{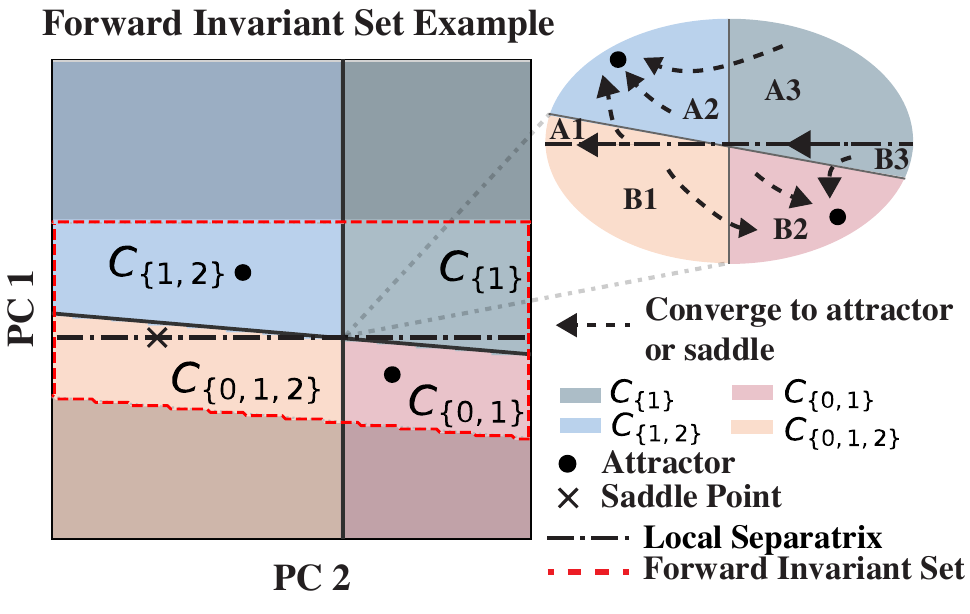}
    \caption{Illustrative example of a forward invariant set in 4D CSTLN projected onto 2D via PCA. The forward invariant set (contained by the red dashed line) contains two ROAs corresponding to two different attractors. Two ROAs are separated by the black dashed line.}
    \label{fig:cell}
\end{figure}
Figure~\ref{fig:cell} shows a forward-invariant set for a 4-dimensional CSTLN with \(\epsilon=0.7\), \(\delta=2.5\), and \(c=1\). For visualization, the 4D space is projected into 2D via principal component analysis (PCA). In this example, regions \(B1\) to \(B3\) form the ROA of the attractor in \(C_{\{0,1\}}\), whereas regions \(A1\) to \(A3\) form to the ROA of the attractor in \(C_{\{1,2\}
}\). Notice that some attractors, such as the one in \(C_{\{1,2\}}\), have both left (from $\mathcal{R}_{\mathrm{FI}}^{(1)}$) and right (from $\mathcal{R}_{\mathrm{FI}}^{(2)}$) ROA polyhedra, which together form the total ROA estimate. With these approximations, we next show that Theorem~\ref{thm:local_ROA} leads to a LP for computing a larger noise bound after the learning.

\begin{figure*}[ht] 
    \centering
    \includegraphics[width=0.99\linewidth]{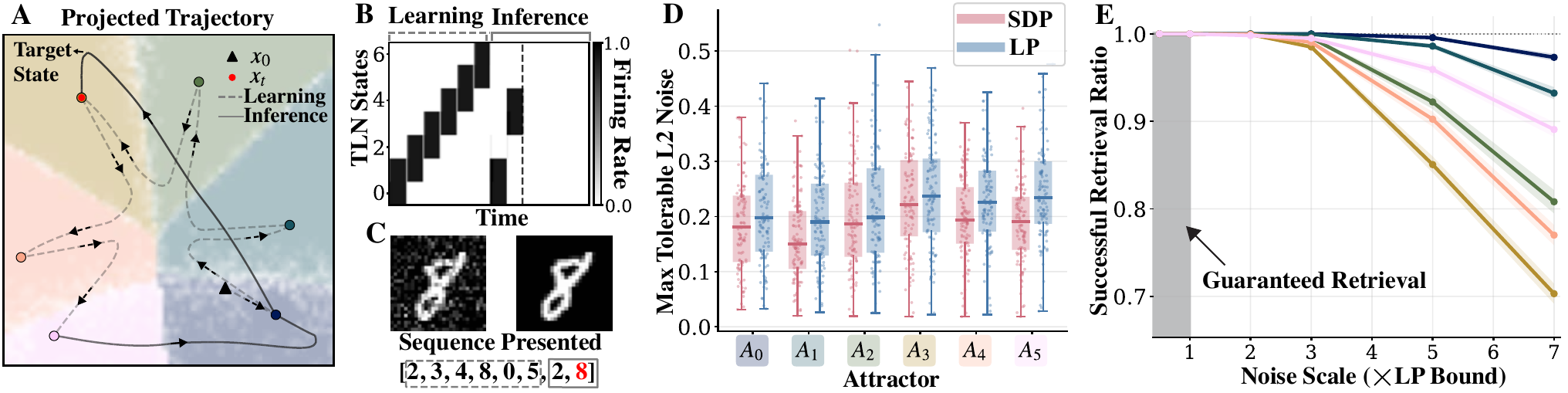}
    \caption{Simulation of 7-dimensional TLN-based memory system and noise robustness analysis on the MNIST dataset. A. Trajectories visualized using projection into 2D space. Dashed lines show trajectories during learning; solid lines show trajectories during inference. The background colors show the true ROA from simulations. B. TLN firing rates over time. The black dashed line marks the end of the second inference pattern (image with number 8). C. The system reconstructs the noisy input pattern. D. Comparison of input-noise robustness for the two methods over 108 encoders and decoders corresponding to different learning sequences randomly sampled from MNIST. E. Empirical simulations show that the system begins to fail to converge to the correct attractor when the noise level exceeds approximately twice the bound found by the LP method, where the bound is taken as the median value in panel D.}
    \label{fig:simulation}
\end{figure*}

\begin{theorem}[LP-based robustness certification]\label{thm:LP}
Let \(P^{(m)}\) be a learned pattern associated with an attractor
\(x^{*,m}\in C_{\{i,i+1\}}\) for some \(i\in\{1,\dots,n-1\}\). 
Let \(\mathcal P_{m}^{\mathrm L}\) and \(\mathcal P_{m}^{\mathrm R}\) denote the left and right ROA polyhedra for \(x^{*,m}\) from $\mathcal{R}_{\mathrm{FI}}^{(i)}$ and $\mathcal{R}_{\mathrm{FI}}^{(i+1)}$, respectively\footnote{If \(x^{*,m}\) is a boundary attractor, that is, in \(C_{\{0,1\}}\) or \(C_{\{n-1,n\}}\), It has only one ROA polyhedral approximation.}. For \(s\in\{\mathrm L,\mathrm R\}\), write
\[
\mathcal P_m^s = \{x\in\mathbb R^n : A_m^s x \le b_m^s\},
\]
where \(A_m^s\in\mathbb{R}^{M_m^s\times n}\) and \(b_m^s\in\mathbb{R}^{M_m^s}\) collect the constraints defining the forward-invariant set in Theorem~\ref{thm:forward_invariant} together with the separating half-space from Theorem~\ref{thm:local_ROA}, written below:
\begin{subequations}\label{eq:poly_def}
\begin{align}
&\alpha_j \le c + 2(-1+\epsilon)\beta_j, \label{eq:poly_alpha_1}\\
&(-1+\epsilon)\alpha_j + c \le \beta_j, \label{eq:poly_alpha_2}\\
&\alpha_j \ge \frac{c}{1+\delta}, \label{eq:poly_alpha_3}\\
&\phi_m^s(x) \le 0, \quad x\in\mathcal P_m^s,\ \text{if } s=\mathrm L, \label{eq:poly_halfspace_L}\\
&\phi_m^s(x) \ge 0, \quad x\in\mathcal P_m^s,\ \text{if } s=\mathrm R, \label{eq:poly_halfspace_R}
\end{align}
\end{subequations}
with \(j=i\) for \(s=\mathrm L\) and \(j=i+1\) for \(s=\mathrm R\).
Assume there exists \(r_m^s>0\) such that the following linear program is feasible:
\begin{subequations}\label{eq:poly_lp}
\begin{align}
\underset{r_m^s}{\mathrm{maximize}} \quad & r_m^s
\label{eq:poly_lp_obj}\\
\text{subject to}\quad
& r_m^s \le
\frac{b_{m,\ell}^{s}-a_{m,\ell}^{s\top}W_E^\top P^{(m)}}
{\|a_{m,\ell}^{s\top}W_E^\top\|_2},
\; \ell=1,\dots,M_m^s,
\label{eq:poly_noise_bound}
\end{align}
\end{subequations}
where \(a_{m,\ell}^{s\top}\) is the \(\ell\)-th row of \(A_m^s\) and \(b_{m,\ell}^{s}\) is the \(\ell\)-th entry of \(b_m^s\). Furthermore, assume the input pattern to the memory system is the previously learned pattern $P^{(m)}$ corrupted with additive noise $\eta$. Then, as long as  \(\|\eta\|_2\le \min (r_m^R, r_m^L)\), the encoder will map the corrupted input pattern to a target state $x_{\mathrm{tar}}$ which lies in the ROA of $x^{*,m}$ (i.e., \(\mathcal P_{m}^{\mathrm L}\) $\cup$ \(\mathcal P_{m}^{\mathrm R}\)).
\end{theorem}

\begin{proof}
For $s \in \{\mathrm{L}, \mathrm{R} \}$, consider the attractor \(x^{*,m}\) and the polyhedron \(\mathcal P_m^s=\{x:A_m^s x\le b_m^s\}\). To guarantee \(x_\mathrm{tar}\in\mathcal P_m^s\), it suffices to enforce
\[
a_{m,l}^{s\top}x_\mathrm{tar} \le b_{m,l}^s,
\qquad l=1,\dots,M_m^s.
\]
Substituting \(x_{\mathrm{tar}}=W_E^\top(P^{(m)}+\eta)\) and using Cauchy--Schwarz gives, for any \(\|\eta\|_2\le r_m^s\),
\[
a_{m,\ell}^{s\top}W_E^\top P^{(m)}+\left\|a_{m,\ell}^{s\top}W_E^\top\right\|_2\,r_m^s \le b_{m,\ell}^s,
\qquad \ell=1,\dots,M_m^s.
\]
which yields~\eqref{eq:poly_noise_bound}. Therefore, \(x_{\mathrm{tar}}\in\mathcal P_m^s\).
\end{proof}

\section{Simulations} \label{sec:simulations}
In this section, we present an example of learning and inference, followed by a robustness analysis. Input patterns are drawn from the MNIST dataset~\cite{lecun2002gradient}. We use a 7-dimensional CSTLN with \(\epsilon=0.9\), \(\delta=2\), and \(c=1\). Fig.~\ref{fig:simulation}A shows the trajectory of TLN during learning (dashed lines) and inference (solid lines). The small bend in the solid trajectory marks the moment when control is turned off at the target state and the system returns to the autonomous TLN dynamics. Background colors show the true ROA (computed through sampling). Figure~\ref{fig:simulation}B Shows the corresponding TLN firing rates over time, where each chunk corresponds to an attractor in the TLN. Fig.~\ref{fig:simulation}C shows the successful reconstruction of a noisy input pattern. Next, we evaluate robustness using the two methods described in Theorems~\ref{thm:SDP} and~\ref{thm:LP}. Figure~\ref{fig:simulation}D summarizes the results on 108 encoders and decoders corresponding to different learning sequences, each containing 6 patterns. As shown in Fig.~\ref{fig:simulation}D, the LP method in Theorem~\ref{thm:LP} outperforms the SDP method in Theorem~\ref{thm:SDP} by yielding a larger \(l_2\)-norm noise bound. Figure~\ref{fig:simulation}E further shows that, on the same learning examples, the empirical worst-case noise bound guaranteeing convergence to the correct attractor is approximately twice the bound (value is taken as the median in Figure~\ref{fig:simulation}D) found through LP method. Beyond this level, the TLN may converge to a wrong attractor and retrieve the wrong memory.

\section{Conclusion}
In this work, we presented a novel associative memory system with a TLN latent space, dynamical controller, linear encoder, and linear decoder. We provided ROA and robustness analyses for this system, and validate our results in simulation.
Future work will explore the generalization of this model to hetero-associative memory and motor tasks.

\bibliographystyle{IEEEtran}
\bibliography{ref}
\end{document}